\RequirePackage{fix-cm}
\documentclass[smallextended]{svjour3}       % onecolumn (second format)
\smartqed  % flush right qed marks, e.g. at end of proof
\usepackage{graphicx}
\usepackage{amsfonts}
\usepackage{amssymb,epic, graphicx}
\usepackage[cmex10]{amsmath}
\usepackage{array}
\usepackage{color}
\usepackage{mdwmath}
\usepackage{mdwtab}
\usepackage{eqparbox}
\usepackage[tight,footnotesize]{subfigure}
\usepackage{algorithm}
\usepackage{algorithmic}
\allowdisplaybreaks
\begin{document}

\title{Complete Weight Enumerators of a Family of Three-Weight Linear Codes%\thanks{Grants or other notes
}

\titlerunning{The CWE of Three-Weight Linear Codes}        % if too long for running head

\author{Shudi Yang      \and
        Zheng-An Yao
        }


\institute{S.D. Yang \at
              Department of Mathematics,
Sun Yat-sen University, Guangzhou 510275 and School of Mathematical
Sciences, Qufu Normal University, Shandong 273165, P.R.China \\
              Tel.: +86-15602338023\\
              \email{yangshd3@mail2.sysu.edu.cn}           %  \\
              \and
           Z.-A. Yao \at
               Department of Mathematics,
Sun Yat-sen University, Guangzhou 510275, P.R. China\\
              \email{mcsyao@mail.sysu.edu.cn}  }
            \date{Received: date / Accepted: date}

\maketitle

\begin{abstract}
Linear codes have been an interesting topic in both theory and
practice for many years. In this paper, for an odd prime $p$, we present the explicit complete weight enumerator of a family of $p$-ary linear codes constructed with defining set. The weight enumerator is an immediate result of the complete weight
enumerator,
which shows that the codes proposed in this paper are three-weight linear
codes. Additionally, all nonzero codewords are minimal and thus they are suitable for
secret sharing.

\end{abstract}

%\vspace{1\baselineskip}
%\noindent$\displaystyle \mathbf{Mathematics~~ Subject~~ Classification~~~} $    11T71$ \cdot$94B15

\section{Introduction}\label{sec:intro}

Throughout this paper, let $p$ be an odd prime and $r=p^m$ for a positive integer $m\geq2$. Denote by $\mathbb{F}_r$ a
finite field with $r$ elements. An $[n, \kappa, \delta]$ linear code
$C$ over $\mathbb{F}_p$ is a $\kappa$-dimensional subspace of
$\mathbb{F}_p^n$ with minimum distance $\delta$~\cite{ding2014differencesets,macwilliams1977theory}.

Let $A_i$ denote the number of codewords with Hamming weight
$i$ in a linear code $C$ of length $n$. The (ordinary) weight enumerator of  $C$  is defined by
$$A_0+A_1z+A_2z^2+\cdots+A_nz^n,$$
where $A_0=1$. The sequence $(A_0,A_1,A_2,\cdots,A_n)$ is called the (ordinary) weight
distribution of the code $C$.

The complete weight enumerator of a code $C$ over $\mathbb{F}_p$ enumerates the codewords according to the number of symbols of each kind contained in each codeword. Denote
elements of the field
by $\mathbb{F}_p=\{w_0,w_1,\cdots,w_{p-1}\}$, where $w_0=0$.
Also let $\mathbb{F}_p^*$ denote $\mathbb{F}_p\backslash\{0\}$.
For a codeword $\mathsf{c}=(c_0,c_1,\cdots,c_{n-1})\in \mathbb{F}_p^n$, let $w[\mathsf{c}]$ be the
complete weight enumerator of $\mathsf{c}$, which is defined as
$$w[\mathsf{c}]=w_0^{k_0}w_1^{k_1}\cdots w_{p-1}^{k_{p-1}},$$
where $k_j$ is the number of components of $\mathsf{c}$ equal to $w_j$, $\sum_{j=0}^{p-1}k_j=n$.
The complete weight enumerator of the code $C$ is then
$$\mathrm{CWE}(C)=\sum_{\mathsf{c}\in C}w[\mathsf{c}].$$

The weight distribution of a linear code has attracted a lot of interests for many years and we
refer the reader to~\cite{ding2011,ding2013hamming,dinh2015recent,feng2008weight,feng2012cyclic,li2014weight,luo2008weight,sharma2012weight,vega2012weight,wang2012weight,yu2014weight,yuan2006weight,zheng2015weight,zhou2014class} and references therein for an overview of the related researches. It is not difficult to see that the complete weight enumerators are just the (ordinary) weight
enumerators for binary linear codes. While for nonbinary linear codes, the weight enumerators can be obtained from their complete weight enumerators.

 The information of the complete weight enumerator of a linear code is of vital use both
in theory and in practice. For instance, Blake and Kith investigated the complete weight enumerator of Reed-Solomon codes and showed that they could be helpful in soft decision decoding~\cite{Blake1991,kith1989complete}. In~\cite{helleseth2006}, the study of the monomial and quadratic bent functions was related to the complete weight enumerators of linear codes. It was illustrated by Ding $et~al.$~\cite{ding2007generic,Ding2005auth} that complete weight enumerators can be applied to calculate the deception probabilities of certain authentication codes. In~\cite{chu2006constantco,ding2008optimal,ding2006construction}, the authors studied the complete weight enumerators of some constant
composition codes and presented some families of optimal constant composition codes.

However, it is extremely difficult to evaluate the complete
weight enumerators of linear codes in general and there is little information on this topic in literature besides the above mentioned~\cite{Blake1991,chu2006constantco,ding2008optimal,ding2006construction,kith1989complete}.
Kuzmin and Nechaev investigated the
generalized Kerdock code and related linear codes over Galois rings and determined their complete weight enumerators in~\cite{kuzmin1999complete} and~\cite{kuzmin2001complete}. Very recently, the authors in~\cite{li2015complete} obtained the complete weight enumerators of some cyclic codes by using Gauss sums. Li $et~al.$~\cite{LiYang2015cwe} treated some linear codes by using exponential sums and Galois theory. In this paper, we shall determine the complete weight enumerators of a class of linear codes over finite fields.

Set $\bar{D}=\{d_1,d_2,\cdots,d_n\}\subseteq \mathbb{F}_{r}$. Denote by $\mathrm{Tr}$ the trace function from $\mathbb{F}_{r}$ to $\mathbb{F}_{p}$. A linear code associated with $\bar{D}$ is defined by
\begin{equation*}\label{def:CD'}
    C_{\bar{D}}=\{(\mathrm{Tr}(ad_1),\mathrm{Tr}(ad_2),\cdots,\mathrm{Tr}(ad_n)):
       a\in \mathbb{F}_{r}\}.
\end{equation*}
Then $\bar{D}$ is called the defining set of this code $C_{\bar{D}}$~(see~\cite{ding2015twodesign,dingkelan2014binary,ding2015twothree} for details).

It should be noted that the authors in~\cite{ding2015twodesign,dingkelan2014binary} and~\cite{ding2015twothree} gave the definitions of the code $C_{\bar{D}}$ and the defining set $\bar{D}$. The authors in~\cite{dingkelan2014binary} established binary linear codes $C_{\bar{D}}$ with three weights. In~\cite{ding2015twodesign}, Ding proposed the general construction of the linear codes and determined their weights especially for three specific codes. The authors in~\cite{ding2015twothree} presented the defining set $\bar{D}=\{x\in \mathbb{F}_{r}^*:\mathrm{Tr}(x^{2})=0\}$ to construct a class of linear codes $C_{\bar{D}}$ with
two or three
 weights and investigated their applications in secret sharing, and furthermore, their complete weight enumerators were established by Yang and Yao~\cite{yang2015linear}.

In this paper, let $r=p^m$. The defining set $D$ is given by
\begin{equation}\label{def:D}
     D=\{x\in \mathbb{F}_{r}^*:\mathrm{Tr}(x)=0\},
\end{equation}
and let \begin{equation}\label{def:CD}
    C_{D}=\{(\mathrm{Tr}(ax^2))_{x\in D}:
       a\in \mathbb{F}_{r}\}.
\end{equation}

For this kind of linear codes, we will study their complete weight enumerators, and their weight enumerators as well. We should mention that the main idea of this paper indeed comes from~\cite{ding2015twodesign,ding2015twothree}. As will be proved later,
they are linear codes with three weights. In addition, each nonzero codeword of $C_D$ constructed in this paper is minimal if $m\geq4$, which indicates that the linear codes can be of use in secret sharing schemes~\cite{carlet2005linear} with nice access structures. The reader is referred to~\cite{ding2015twothree} for more information about this application.

The main results of this paper are given below.

\begin{theorem}\label{thcwe:CD}
Let $N_a(\rho)=\#\{x\in \mathbb{F}_{r}^*:\mathrm{Tr}(x)=0~ and~ \mathrm{Tr}(ax^2)=\rho \}$ with $\rho\in \mathbb{F}_{p}$ and $a\in \mathbb{F}_{r}$. Then the code
$C_D$ of \eqref{def:CD} has parameters
$[ p^{m-1}-1,m]$ and its complete weight enumerator is
given in Table~\ref{CWE:m even simplified}
if $m$ is even and Table~\ref{CWE:m odd} if $m$ is odd.

\begin{table}[htbp]
\tabcolsep 2mm \caption{{Complete weight enumerator of the
code $C_D$ if $m$ is even}}\label{CWE:m even simplified}
\begin{center}\begin{tabular}{lll}
\hline\noalign{\smallskip}
$N_a(0)$ & $N_a(\rho) (\rho\neq0)$ & Frequency  \\
\noalign{\smallskip}\hline\noalign{\smallskip}
$p^{m-2}-1$     & $p^{m-2}+\bar{\eta}(\rho)p^{\frac{m-2}{2}}$   &
$\frac{p^m-p^{m-1}}{2}$ \\
$p^{m-2}-1$     & $p^{m-2}-\bar{\eta}(\rho)p^{\frac{m-2}{2}}$   &
$\frac{p^m-p^{m-1}}{2}$ \\
$p^{m-2}-1-(p-1)p^{\frac{m-2}{2}}$     & $p^{m-2}+p^{\frac{m-2}{2}}$   & $\frac{1}{2}(p^{\frac{m}{2}}+1)(p^{\frac{m-2}{2}}-1)$ \\
$p^{m-2}-1+(p-1)p^{\frac{m-2}{2}}$     & $p^{m-2}-p^{\frac{m-2}{2}}$   & $\frac{1}{2}(p^{\frac{m}{2}}-1)(p^{\frac{m-2}{2}}+1)$ \\
$p^{m-1}-1$         &  0           & 1   \\
\noalign{\smallskip}\hline
\end{tabular}
\end{center}
where {$\bar{\eta}$ denotes the quadratic character of $\mathbb{F}_{p}$}.
\end{table}

%\begin{table}[htbp]
%\tabcolsep 2mm \caption{Complete weight enumerator of the
%code $C_D$ if $m$ is even}\label{CWE:m even}
%\begin{center}\begin{tabular}{lll}
%\hline\noalign{\smallskip}
%$N_a(0)$ & $N_a(\rho) (\rho\neq0)$ & Frequency  \\
%$p^{m-2}-1$     & $p^{m-2}+\bar{\eta}(\rho)p^{\frac{m-2}{2}}$   &
%$\frac{p^m-p^{m-1}}{2}$ \\
%$p^{m-2}-1$     & $p^{m-2}-\bar{\eta}(\rho)p^{\frac{m-2}{2}}$   &
%$\frac{p^m-p^{m-1}}{2}$ \\
%$p^{m-2}-1-(-1)^{\frac{p-1}{2}\frac{m}{2}}(p-1)p^{\frac{m-2}{2}}$     & $p^{m-2}+(-1)^{\frac{p-1}{2}\frac{m}{2}}p^{\frac{m-2}{2}}$   & $\frac{p^m-1}{2p}+\frac{p-1}{p}\eta_0^{(2,r)}$ \\
%$p^{m-2}-1+(-1)^{\frac{p-1}{2}\frac{m}{2}}(p-1)p^{\frac{m-2}{2}}$     & $p^{m-2}-(-1)^{\frac{p-1}{2}\frac{m}{2}}p^{\frac{m-2}{2}}$   & $\frac{p^m-1}{2p}+\frac{p-1}{p}\eta_1^{(2,r)}$ \\
%$p^{m-1}-1$         &  0           & 1   \\
%\end{tabular}
%\end{center}
%where $\eta_0^{(2,r)}$ and $\eta_1^{(2,r)} $ are Gaussian
%periods of order $2$
%{and are given in Lemma~\ref{lemN=2},}
%and {$\bar{\eta}$ denotes the quadratic character of $\mathbb{F}_{p}$}.
%\end{table}

\begin{table}[htbp]
\tabcolsep 2mm \caption{Complete weight enumerator of the
code $C_D$ if $m$ is odd}\label{CWE:m odd}
\begin{center}\begin{tabular}{lll}
\hline\noalign{\smallskip}
$N_a(0)$ & $N_a(\rho) (\rho\neq0)$ & Frequency  \\
\noalign{\smallskip}\hline\noalign{\smallskip}
$p^{m-2}-1$     & $p^{m-2}+\bar{\eta}(\rho)p^{\frac{m-1}{2}}$   &
$\frac{p^{m-1}-1}{2}$ \\
$p^{m-2}-1$     & $p^{m-2}-\bar{\eta}(\rho)p^{\frac{m-1}{2}}$   &
$\frac{p^{m-1}-1}{2}$ \\
$p^{m-2}+(p-1)p^{\frac{m-3}{2}}$     & $p^{m-2}-p^{\frac{m-3}{2}}$   &
$\frac{p-1}{2}\left(p^{m-1}+p^{\frac{m-1}{2}}\right)$  \\
$p^{m-2}-(p-1)p^{\frac{m-3}{2}}$     & $p^{m-2}+p^{\frac{m-3}{2}}$   &
$\frac{p-1}{2}\left(p^{m-1}-p^{\frac{m-1}{2}}\right)$  \\
$p^{m-1}-1$         &  0           & 1   \\
\noalign{\smallskip}\hline
\end{tabular}
\end{center}
where {$\bar{\eta}$ denotes the quadratic character of $\mathbb{F}_{p}$.}
\end{table}

\end{theorem}

\begin{theorem}\label{wt:CD} The weight distribution of
the code $C_D$ of \eqref{def:CD} is
{given in} Table~\ref{wt:m even simplified}
if $m$ is even and Table~\ref{wt:m odd} if $m$ is odd.

\begin{table}[htbp]
\tabcolsep 2mm \caption{{The weight distribution of $C_D$ if $m$ is even}}\label{wt:m even simplified}
\begin{center}\begin{tabular}{ll}
\hline\noalign{\smallskip}
  Weight $i$   & Frequency  $A_i$            \\
    \noalign{\smallskip}\hline\noalign{\smallskip}
   $(p-1)p^{m-2}$ & $p^m-p^{m-1}$     \\
   $(p-1)(p^{m-2}+p^{\frac{m-2}{2}})$   &$\frac{1}{2}(p^{\frac{m}{2}}+1)(p^{\frac{m-2}{2}}-1)$   \\
   $(p-1)(p^{m-2}-p^{\frac{m-2}{2}})$   &$\frac{1}{2}(p^{\frac{m}{2}}-1)(p^{\frac{m-2}{2}}+1)$   \\
   0         &  1        \\
   \noalign{\smallskip}\hline
  \end{tabular}
\end{center}

\end{table}

%\begin{table}[htbp]
%\tabcolsep 2mm \caption{The weight distribution of $C_D$ if $m$ is even}\label{wt:m even}
%\begin{center}\begin{tabular}{ll}
%\hline\noalign{\smallskip}
%  Weight $i$   & Frequency  $A_i$            \\
%    \noalign{\smallskip}\hline\noalign{\smallskip}
%  $(p-1)p^{m-2}$ & $p^m-p^{m-1}$     \\
%   $(p-1)(p^{m-2}+(-1)^{\frac{p-1}{2}\frac{m}{2}}p^{\frac{m-2}{2}})$   &$\frac{p^m-1}{2p}+\frac{p-1}{p}\eta_0^{(2,r)} $   \\
%   $(p-1)(p^{m-2}-(-1)^{\frac{p-1}{2}\frac{m}{2}}p^{\frac{m-2}{2}})$   &$\frac{p^m-1}{2p}+\frac{p-1}{p}\eta_1^{(2,r)} $   \\
%   0         &  1        \\
%   \noalign{\smallskip}\hline
%  \end{tabular}
%\end{center}
%where $\eta_0^{(2,r)} $ and $\eta_1^{(2,r)} $ are Gaussian
%periods of order $2$ and are given in Lemma~\ref{lemN=2}.
%\end{table}

\begin{table}[htbp]
\tabcolsep 2mm \caption{The weight distribution of $C_D$ if $m$ is odd}\label{wt:m odd}
\begin{center}\begin{tabular}{ll}
\hline\noalign{\smallskip}
  Weight $i$   & Frequency  $A_i$            \\
    \noalign{\smallskip}\hline\noalign{\smallskip}
  $(p-1)p^{m-2}$ & $p^{m-1}-1$     \\
   $(p-1)(p^{m-2}-p^{\frac{m-3}{2}})$   &
   $\frac{p-1}{2}\left(p^{m-1}+p^{\frac{m-1}{2}}\right) $   \\
   $(p-1)(p^{m-2}+p^{\frac{m-3}{2}})$   &
   $\frac{p-1}{2}\left(p^{m-1}-p^{\frac{m-1}{2}}\right)$   \\
   0         &  1        \\
   \noalign{\smallskip}\hline
  \end{tabular}
\end{center}
\end{table}

\end{theorem}

\begin{example}
(i) Let $(p,m)=(3,4)$. Then by Theorems~\ref{thcwe:CD} and~\ref{wt:CD}, the code $C_D$ has parameters $[26,4,12]$ with complete
weight enumerator
$$w_0^{26}+16w_0^{14}w_1^{6}w_2^{6}+27w_0^{8}w_1^{6}w_2^{12}+27w_0^{8}w_1^{12}w_2^{6}+10w_0^{2}w_1^{12}w_2^{12},$$
and weight enumerator
$$1+16z^{12}+54z^{18}+10z^{24}.$$
These results coincide with numerical computation by Magma.

(ii) Let $(p,m)=(5,3)$. Then by Theorems~\ref{thcwe:CD} and~\ref{wt:CD}, the code $C_D$ has parameters $[24,3,16]$ with complete
weight enumerator
$$w_0^{24}+60w_0^{8}w_1^{4}w_2^{4}w_3^{4}w_4^{4}+12w_0^{4}w_1^{10}w_4^{10}
+12w_0^{4}w_2^{10}w_3^{10}+40w_1^{6}w_2^{6}w_3^{6}w_4^{6}.$$
and weight enumerator
$$1+60z^{16}+24z^{20}+40z^{24},$$
These results are confirmed by Magma.

\end{example}

The remainder of this paper is organized as follows. Section
\ref{sec:mathtool}
{recalls} some definitions and results on
Gaussian periods and Gauss sums over finite fields. Section~\ref{sec:proof}
is devoted to the proof of Theorem~\ref{thcwe:CD}. Section~\ref{sec:minimal} shows that every nonzero
codeword of the code is minimal. Section~\ref{sec:conclusion} concludes this paper.

\section{Mathematical foundations}\label{sec:mathtool}

We start with cyclotomic classes and Gaussian periods over finite fields. Recall that $r=p^m$. Let $\alpha$ be a fixed primitive element of $\mathbb{F}_r$ and
{$r-1=sN$, where $s$, $N$ are two positive integers with $s>1$ and $N>1$}.
Define $C_i^{(N,r)}=\alpha
^i\langle\alpha^N\rangle$ for $i=0,1,\cdots,N-1$, where
$\langle\alpha^N\rangle$ denotes the subgroup of $\mathbb{F}_r^*$
generated by $ \alpha^N$. The cosets $C_i^{(N,r)}$ are called the \emph{cyclotomic classes} of order $N$
in $\mathbb{F}_r$. Obviously, $\#C_i^{(N,r)}=\frac{r-1}{N}$.

 Set
$\zeta_p=\exp\left(\frac{2\pi\sqrt{-1}}{p}\right)$. The Gaussian
periods of order $N$ are defined by
 $$\eta_{i}^{(N,r)}=\sum_{x\in C_i^{(N,r)}}\chi_1(x).$$
 Here $\chi_1(x)=\zeta_p^{\text{Tr}(x)}$ is the
 canonical additive character of $\mathbb{F}_r$, where $\text{Tr}$ is the trace function from $\mathbb{F}_r$
to $\mathbb{F}_p$ defined by
\begin{equation*}
\mathrm{Tr}(x)=\sum^{m-1}_{i=0}x^{p^i},~~x\in \mathbb{F}_{r}.
\end{equation*}

The orthogonal property of {canonical additive character},
which can be easily checked, is {given by}
\begin{equation*}\label{eq:ortho}
\sum_{x\in \mathbb{F}_r}\chi_1(ax)
=\left\{\begin{array}{lll}
r                   &&~~\mbox{if}~~a=0, \\
0                   &&~~\mbox{if}~~a\in\mathbb{F}_r^*.\\
                                            \end{array}
                                       \right.
\end{equation*}

Some results on Gaussian periods are given below~\cite{myerson1981period}.

\begin{lemma}\emph{\cite{myerson1981period}}\label{lemN=2}
Let $r=p^m$. When $N=2$, the Gaussian periods are given by
\begin{equation*}
\eta_0^{(2,r)}=\left\{\begin{array}{lll} \frac{-1+(-1)^{ m-1}\sqrt{r}}{2} & & ~~if~~ p\equiv1\pmod 4,\\
\frac{-1+(-1)^{ m-1}{(\sqrt{-1})}^{ m}\sqrt{r}}{2} & & ~~if ~~p\equiv3\pmod 4,\\
\end{array}
 \right.
\end{equation*}
and $\eta_1^{(2,r)}=-1-\eta_0^{(2,r)}$.
\end{lemma}

For later use, we introduce Gauss sums in the following.
{Let $\eta$ be the quadratic character of $\mathbb{F}_r$~\cite{ding2015twothree}.}
%such that
%\begin{eqnarray*}
%\eta(x)=\left\{\begin{array}{lll}1,
%&&if ~~x ~~is~~ a~~ square~~ in~~ \mathbb{F}_{r}^*,\\
%-1,&&if ~~x ~~is~~ a~~ nonsquare~~ in~~ \mathbb{F}_{r}^*,\\
%0,&&if ~~x=0.\\
%\end{array}
%\right.
%\end{eqnarray*}
The quadratic Gauss sum $G(\eta,\chi_1)$ over $\mathbb{F}_{r}$ is defined by
\begin{eqnarray*}
G(\eta,\chi_1)=\sum_{x\in\mathbb{F}_r^*}\eta(x)\chi_1(x)=\sum_{x\in\mathbb{F}_r}\eta(x)\chi_1(x),
\end{eqnarray*}
and the quadratic Gauss sum $G(\bar{\eta},\bar{\chi}_1)$ over $\mathbb{F}_{p}$ is defined by
\begin{eqnarray*}
G(\bar{\eta},\bar{\chi}_1)=\sum_{x\in\mathbb{F}_{p}^*}\bar{\eta}(x)\bar{\chi}_1(x)
                          =\sum_{x\in\mathbb{F}_{p}}\bar{\eta}(x)\bar{\chi}_1(x),
\end{eqnarray*}
where $\bar{\eta}$ and $\bar{\chi}_1$ are the quadratic character and
{canonical additive character} of $\mathbb{F}_{p}$, respectively.

The following lemmas will be needed in the sequel.
\begin{lemma}(See Theorems 5.15~\cite{lidl1983finite})\label{lm:gauss sum}
With the symbols and notation above, we have
\begin{eqnarray*}\label{eq:Gausspm}
G(\eta,\chi_1)=(-1)^{m-1}(\sqrt{-1})^{\frac{(p-1)^2}{4}m}\sqrt{r},
\end{eqnarray*}
where $r=p^m$, and
\begin{eqnarray*}\label{eq:Gaussp}
G(\bar{\eta},\bar{\chi}_1)=(\sqrt{-1})^{\frac{(p-1)^2}{4}}\sqrt{p}.
\end{eqnarray*}

\end{lemma}

\begin{lemma}(See Theorem 5.33 of~\cite{lidl1983finite})\label{lm:expo sum}
{Let $r=p^m$ and $f(x)=a_2x^2+a_1x+a_0\in \mathbb{F}_{r}[x]$ with
$a_2\neq0$.}
Then
\begin{eqnarray*}\label{eq:expo sum}
\sum_{x\in
\mathbb{F}_{r}}\chi_1(f(x))=\chi_1(a_0-a_1^2(4a_2)^{-1})\eta(a_2)G(\eta,\chi_1).
\end{eqnarray*}
\end{lemma}

\begin{lemma}(See Lemma 7 of~\cite{ding2015twothree})\label{le:eta}
If $m\geq 2$ is even, then $\eta(y) = 1$ for each $y \in \mathbb{F}_{p}^*$.
If $m$ is odd, then $\eta(y) = \bar{\eta}(y)$ for each $y \in
\mathbb{F}_{p}$.

\end{lemma}

{
\begin{lemma}\label{lem:ti}(\cite{ding2015twothree,li2015complete})
For each $c\in \mathbb{F}_{p}$, let $t_c=\#\{a\in \mathbb{F}_{r}:\mathrm{Tr}(a^2)=c\}$. Then
 \begin{eqnarray*}
 t_c=\left\{\begin{array}{lll}
  p^{m-1} && ~~if~~ m ~~odd~~ and~~ c=0,\\
 p^{m-1}+\bar{\eta}(c)(-1)^{\frac{p-1}{2}\frac{m-1}{2}}p^{\frac{m-1}{2}} && ~~if~~m~~odd~~and~~c\neq0, \\
 p^{m-1}-(-1)^{\frac{p-1}{2}\frac{m}{2}}(p-1)p^{\frac{m-2}{2}} && ~~if~~ m ~~even~~ and~~ c=0,\\
 p^{m-1}+(-1)^{\frac{p-1}{2}\frac{m}{2}}p^{\frac{m-2}{2}} && ~~if~~ m ~~even~~ and~~ c\neq0.
 \end{array}\right.
 \end{eqnarray*}
\end{lemma}
}

\section{The proofs of the main results}\label{sec:proof}

Let notation be as before. Our task of this section is to prove Theorem~\ref{thcwe:CD} depicted in Section~\ref{sec:intro}, while Theorem~\ref{wt:CD} follows immediately from Theorem~\ref{thcwe:CD}. Below we present some auxiliary results before proving the main results of this paper.

\begin{lemma}\label{lem:sum3}
Let $a\in \mathbb{F}_{r}^* $ and $\rho\in \mathbb{F}_{p}^*$. Then we have
\begin{eqnarray*}
  &&\sum_{z\in\mathbb{F}_{p}^*}\sum_{y\in\mathbb{F}_{p}}\sum_{x\in\mathbb{F}_{r}}\zeta_p^{\mathrm{Tr}(azx^2+yx)-z\rho}\\
  &&=\left\{\begin{array}{lll}
   (-1)^{\frac{p-1}{2}\frac{m}{2}}\eta(a)p^{\frac{m+2}{2}} &&if~m~even ~and~\mathrm{Tr}(a^{-1})=0,\\
   -(-1)^{\frac{p-1}{2}\frac{m+2}{2}}\eta(a)\bar{\eta}(\mathrm{Tr}(a^{-1}))\bar{\eta}(\rho)p^{\frac{m+2}{2}} &&if~m~even ~and~\mathrm{Tr}(a^{-1})\neq0,\\
   (-1)^{\frac{p-1}{2}\frac{m-1}{2}}\eta(a)\bar{\eta}(\rho)p^{\frac{m+3}{2}} &&if~m~odd ~and~\mathrm{Tr}(a^{-1})=0,\\
   -(-1)^{\frac{p-1}{2}\frac{m-1}{2}}\eta(a)\bar{\eta}(\mathrm{Tr}(a^{-1}))p^{\frac{m+1}{2}} &&if~m~odd ~and~\mathrm{Tr}(a^{-1})\neq0.\\
\end{array} \right.
 \end{eqnarray*}

\end{lemma}
\begin{proof}
{
It follows from Lemmas~\ref{lm:expo sum} and~\ref{le:eta} that
\begin{eqnarray*}
  &&\sum_{z\in\mathbb{F}_{p}^*}\sum_{y\in\mathbb{F}_{p}}\sum_{x\in\mathbb{F}_{r}}\zeta_p^{\mathrm{Tr}(azx^2+yx)-z\rho}\\
  &&=\sum_{z\in\mathbb{F}_{p}^*}\zeta_p^{-z\rho}\sum_{y\in\mathbb{F}_{p}}\sum_{x\in\mathbb{F}_{r}}\zeta_p^{\mathrm{Tr}(azx^2+yx)}\\
  &&=\sum_{z\in\mathbb{F}_{p}^*}\zeta_p^{-z\rho}\sum_{y\in\mathbb{F}_{p}}
  \chi_1\left(-\frac{y^2}{4az}\right)\eta(az)G(\eta,\chi_1)\\
  &&=\eta(a)G(\eta,\chi_1)\sum_{z\in\mathbb{F}_{p}^*}\zeta_p^{-z\rho}\eta(z)
  \sum_{y\in\mathbb{F}_{p}}\zeta_p^{-\frac{1}{4z}\mathrm{Tr}(a^{-1})y^2}\\
&&=\left\{\begin{array}{lll}
  p\eta(a)G(\eta,\chi_1)\sum_{z\in\mathbb{F}_{p}^*}\zeta_p^{-z\rho}\eta(z)
  &&\mathrm{if}~\mathrm{Tr}(a^{-1})=0\\
   \eta(a)G(\eta,\chi_1)\sum_{z\in\mathbb{F}_{p}^*}\zeta_p^{-z\rho}\eta(z)\bar{\eta}
   \left(-\frac{\mathrm{Tr}(a^{-1})}{4z}\right)G(\bar{\eta},\bar{\chi}_1)
   &&\mathrm{if}~\mathrm{Tr}(a^{-1})\neq0\\
   \end{array} \right.\\
&&=\left\{\begin{array}{lll}
  p\eta(a)G(\eta,\chi_1)\sum_{z\in\mathbb{F}_{p}^*}\zeta_p^{-z\rho}\eta(z)
  &&\mathrm{if}~\mathrm{Tr}(a^{-1})=0\\
   \eta(a)\bar{\eta}(-\mathrm{Tr}(a^{-1}))G(\eta,\chi_1)G(\bar{\eta},\bar{\chi}_1)
   \sum_{z\in\mathbb{F}_{p}^*}\zeta_p^{-z\rho}\eta(z)\bar{\eta}(z)
   &&\mathrm{if}~\mathrm{Tr}(a^{-1})\neq0\\
\end{array} \right.\\
&&=\left\{\begin{array}{lll}
  -p\eta(a)G(\eta,\chi_1)
  &&\mathrm{if}~m~\mathrm{even} ~\mathrm{and}~\mathrm{Tr}(a^{-1})=0,\\
   \eta(a)\bar{\eta}(-\mathrm{Tr}(a^{-1}))\bar{\eta}(-\rho)G(\eta,\chi_1)G(\bar{\eta},\bar{\chi}_1)^2
  &&\mathrm{if}~m~\mathrm{even} ~\mathrm{and}~\mathrm{Tr}(a^{-1})\neq0,\\
    p\eta(a)\bar{\eta}(-\rho)G(\eta,\chi_1)G(\bar{\eta},\bar{\chi}_1)
  &&\mathrm{if}~m~\mathrm{odd}~ \mathrm{and}~\mathrm{Tr}(a^{-1})=0,\\
   -\eta(a)\bar{\eta}(-\mathrm{Tr}(a^{-1}))G(\eta,\chi_1)G(\bar{\eta},\bar{\chi}_1)
  &&\mathrm{if}~m~\mathrm{odd}~ \mathrm{and}~\mathrm{Tr}(a^{-1})\neq0.\\
\end{array} \right.
\end{eqnarray*}

From Lemma~\ref{lm:gauss sum}, we get the desired conclusions.

}

%from the fact that
%\begin{eqnarray*}
%  \left\{\begin{array}{lll}
%   (-1)^{(\frac{p-1}{2})^2\frac{m}{2}}=(-1)^{\frac{p-1}{2}\frac{m}{2}} &&\mathrm{if}~~m~~\mathrm{even},\\   (-1)^{(\frac{p-1}{2})^2\frac{m+1}{2}+\frac{p-1}{2}}=(-1)^{\frac{p-1}{2}\frac{m-1}{2}} &&\mathrm{if}~~m~~\mathrm{odd}. \\
%\end{array} \right.
% \end{eqnarray*}

%\begin{eqnarray*}
% &&=\left\{\begin{array}{lll}
%   (-1)^{(\frac{p-1}{2})^2\frac{m}{2}}\eta(a)p^{\frac{m+2}{2}} &&\mathrm{if}~m~\mathrm{even} ~\mathrm{and}~\mathrm{Tr}(a^{-1})=0,\\
%   -(-1)^{(\frac{p-1}{2})^2\frac{m+2}{2}}\eta(a)\bar{\eta}(\mathrm{Tr}(a^{-1}))\bar{\eta}(\rho)p^{\frac{m+2}{2}} &&\mathrm{if}~m~\mathrm{even} ~\mathrm{and}~\mathrm{Tr}(a^{-1})\neq0,\\
%   (-1)^{(\frac{p-1}{2})^2\frac{m+1}{2}+\frac{p-1}{2}}\eta(a)\bar{\eta}(\rho)p^{\frac{m+3}{2}} &&\mathrm{if}~m~\mathrm{odd} ~\mathrm{and}~\mathrm{Tr}(a^{-1})=0,\\
%   -(-1)^{(\frac{p-1}{2})^2\frac{m+1}{2}+\frac{p-1}{2}}\eta(a)\bar{\eta}(\mathrm{Tr}(a^{-1}))p^{\frac{m+1}{2}} &&\mathrm{if}~m~\mathrm{odd} ~\mathrm{and}~\mathrm{Tr}(a^{-1})\neq0.\\
%\end{array} \right.
%\end{eqnarray*}

\hfill\space$\qed$
 \end{proof}

\begin{lemma}\label{lem:Na rho}
For any $a\in \mathbb{F}_{r}^* $ and any $\rho\in \mathbb{F}_{p}^*$, let
\begin{equation*}
    N_a(\rho)=\#\{x\in \mathbb{F}_{r}^*:\mathrm{Tr}(x)=0 ~~ and~~ \mathrm{Tr}(ax^2)=\rho \}.
\end{equation*}
Then we have
\begin{eqnarray*}
   N_a(\rho)=\left\{\begin{array}{lll}
    &&p^{m-2}+(-1)^{\frac{p-1}{2}\frac{m}{2}}\eta(a)p^{\frac{m-2}{2}} \\
    &&~~~~~~~~~~~~~~~~~~~~~~~~~~~~~~if~~m~~even ~~and~~\mathrm{Tr}(a^{-1})=0,\\
    &&p^{m-2}-(-1)^{\frac{p-1}{2}\frac{m+2}{2}}\eta(a)\bar{\eta}(\mathrm{Tr}(a^{-1}))
    \bar{\eta}(\rho)p^{\frac{m-2}{2}} \\
    &&~~~~~~~~~~~~~~~~~~~~~~~~~~~~~~if~~m~~even ~~and~~\mathrm{Tr}(a^{-1})\neq0,\\
    &&p^{m-2}+(-1)^{\frac{p-1}{2}\frac{m-1}{2}}\eta(a)\bar{\eta}(\rho)p^{\frac{m-1}{2}} \\
    &&~~~~~~~~~~~~~~~~~~~~~~~~~~~~~~if~~m~~odd ~~and~~\mathrm{Tr}(a^{-1})=0,\\
    &&p^{m-2}-(-1)^{\frac{p-1}{2}\frac{m-1}{2}}\eta(a)\bar{\eta}(\mathrm{Tr}(a^{-1}))p^{\frac{m-3}{2}} \\ &&~~~~~~~~~~~~~~~~~~~~~~~~~~~~~~if~~m~~odd ~~and~~\mathrm{Tr}(a^{-1})\neq0.\\
\end{array} \right.
 \end{eqnarray*}
\end{lemma}
\begin{proof}
We point out that
\begin{eqnarray*}
    N_a(\rho)=\#\{x\in \mathbb{F}_{r}:\mathrm{Tr}(x)=0 ~~ and~~ \mathrm{Tr}(ax^2)=\rho \}
\end{eqnarray*} since $\rho\neq 0$.
Hence, for any $a\in \mathbb{F}_{r}^* $ and any $\rho\in \mathbb{F}_{p}^*$, we have
\begin{eqnarray*}
    N_a(\rho)&=&
        p^{-2}\sum_{x\in\mathbb{F}_{r}}\left(\sum_{y\in\mathbb{F}_{p}}\zeta_p^{y\mathrm{Tr}(x)}\right)
        \left(\sum_{z\in\mathbb{F}_{p}}\zeta_p^{z(\mathrm{Tr}(ax^2)-\rho)}\right)\\
    &=&
    p^{-2}\sum_{y\in\mathbb{F}_{p}^*}\sum_{x\in\mathbb{F}_{r}}\zeta_p^{\mathrm{Tr}(yx)}+ p^{-2}\sum_{z\in\mathbb{F}_{p}^*}\sum_{y\in\mathbb{F}_{p}}\sum_{x\in\mathbb{F}_{r}}\zeta_p^{\mathrm{Tr}(azx^2+yx)-z\rho}+p^{m-2}.
 \end{eqnarray*}
 Note that
\begin{equation*}
\sum_{y\in\mathbb{F}_{p}^*}\sum_{x\in\mathbb{F}_{r}}\zeta_p^{\mathrm{Tr}(yx)}=0
\end{equation*}
since
\begin{equation*}
   \sum_{x\in\mathbb{F}_{r}}\zeta_p^{\mathrm{Tr}(yx)}=\left\{\begin{array}{lll}r,&&\mathrm{if}~~y=0,\\
0,&&\mathrm{if}~~y\in\mathbb{F}_{p}^*.\\
\end{array} \right.
 \end{equation*}
The desired conclusions then follow from Lemma~\ref{lem:sum3}.
\hfill\space$\qed$
\end{proof}

In order to calculate the frequency of the
{complete weight enumerator of codewords} in $C_D$, we shall compute
\begin{equation*}
    \#\{a\in \mathbb{F}_{r}^*:\eta(a)=\pm1 ~~ \mathrm{and}~~ \mathrm{Tr}(a^{-1})=0 \}
\end{equation*}
and
\begin{equation*}
   \#\{a\in \mathbb{F}_{r}^*:\eta(a)\bar{\eta}(\mathrm{Tr}(a^{-1}))=\pm1 \},
\end{equation*} which are given in the following three lemmas.
\begin{lemma}\label{lem:ni}
For any $a\in \mathbb{F}_{r}^* $, let
\begin{equation*}
    n_i=\#\{a\in \mathbb{F}_{r}^*:\eta(a)=i ~~ and~~ \mathrm{Tr}(a^{-1})=0 \}, ~~i\in\{1,-1\}.
\end{equation*}
Then, for $m$ being odd, we have
\begin{eqnarray*}
     n_1=   n_{-1}=\frac{p^{m-1}-1}{2},
 \end{eqnarray*}
 and for $m$ being even, we have
\begin{eqnarray*}
  \left\{\begin{array}{lll}
   n_1&=&\frac{r-1}{2p}+\frac{p-1}{p}\eta_0^{(2,r)},\\
   n_{-1}&=&\frac{r-1}{2p}+\frac{p-1}{p}\eta_1^{(2,r)},\\
\end{array} \right.
 \end{eqnarray*} where $\eta_0^{(2,r)}$ and $\eta_1^{(2,r)}$ are given in Lemma~\ref{lemN=2}.
\end{lemma}
\begin{proof}
Observe that $  n_{-1}=n-n_1=p^{m\!-\!1}\!-1-n_1$. Thus, we only focus on $n_1$.

Let $\alpha$ be a fixed primitive element of $\mathbb{F}_r$. Then $\mathbb{F}_r^*=\langle \alpha \rangle$ and $\mathbb{F}_p^*=\langle \alpha^{\frac{r-1}{p-1}} \rangle$.

Note that $\eta(a)=1$ if and only if $a\in C_0^{(2,r)}$.

For the case of $m$ being even, we have $\mathbb{F}_p^* \subseteq C_0^{(2,r)}$ since $2$ divides $\frac{r-1}{p-1}$.

Therefore, we obtain
\begin{eqnarray*}
 n_1&=&\sum_{x\in C_0^{(2,r)}}\frac{1}{p}\sum_{y\in\mathbb{F}_{p}}\zeta_p^{y\mathrm{Tr}(x)}\\
    &=&\frac{r-1}{2p}+\frac{1}{p}\sum_{y\in\mathbb{F}_{p}^*}\sum_{x\in C_0^{(2,r)}}\zeta_p^{\mathrm{Tr}(yx)}\\
    &=&\frac{r-1}{2p}+\frac{p-1}{p}\eta_0^{(2,r)},
  \end{eqnarray*}
where $\eta_0^{(2,r)}$ is given in Lemma~\ref{lemN=2}.

{
Similarly, for the case of $m$ being odd, we have $\frac{r-1}{p-1}\equiv 1\mod2$,
from which the asserted result follows.
}

%For the case of $m$ being odd, we have $\frac{r-1}{p-1}\equiv 1\mod2$.
%
%It follows that
%\begin{eqnarray*}
% n_1&=&\sum_{x\in C_0^{(2,r)}}\frac{1}{p}\sum_{y\in\mathbb{F}_{p}}\zeta_p^{y\mathrm{Tr}(x)}\\
%    &=&\frac{r-1}{2p}+\frac{1}{p}\sum_{y\in\mathbb{F}_{p}^*}\sum_{x\in C_0^{(2,r)}}\zeta_p^{\mathrm{Tr}(yx)}\\
%    &=&\frac{r-1}{2p}+\frac{1}{p}\left(\frac{p-1}{2}\eta_0^{(2,r)}+\frac{p-1}{2}\eta_1^{(2,r)}\right)\\
%    &=&\frac{r-1}{2p}-\frac{p-1}{2p}\\
%    &=&\frac{p^{m-1}-1}{2}.\\
%  \end{eqnarray*}
  \hfill\space$\qed$
\end{proof}

\begin{lemma}\label{lem:nij}
For any $a\in \mathbb{F}_{r}^* $, let
\begin{equation*}
    n_{i,j}=\#\{a\in \mathbb{F}_{r}^*:\eta(a)=i ~~ and~~ \bar{\eta}(\mathrm{Tr}(a^{-1}))=j \}, ~~i,j\in\{1,-1\}.
\end{equation*}
Then, for $m$ being even, we have
\begin{eqnarray*}
     n_{1,1}=   n_{1,-1}=\frac{p-1}{4}\left(p^{m-1}+(-1)^{\frac{p-1}{2}\frac{m}{2}}p^{\frac{m-2}{2}}\right),
 \end{eqnarray*}
 and for $m$ being odd, we have
\begin{eqnarray*}
  \left\{\begin{array}{lll}
   n_{1,1}&=&\frac{p-1}{4}\left(p^{m-1}+(-1)^{\frac{p-1}{2}\frac{m-1}{2}}p^{\frac{m-1}{2}}\right),\\
   n_{1,-1}&=&\frac{p-1}{4}\left(p^{m-1}-(-1)^{\frac{p-1}{2}\frac{m-1}{2}}p^{\frac{m-1}{2}}\right).\\
\end{array} \right.
 \end{eqnarray*}
\end{lemma}
\begin{proof}
We note the following fact:
\begin{eqnarray*}
n_{1,j}&&=\#\{a\in \mathbb{F}_{r}^*:a\in C_0^{(2,r)}, \bar{\eta}(\mathrm{Tr}(a^{-1}))=j\}\\
&&=\frac{1}{2}\#\{a\in \mathbb{F}_{r}^*:\bar{\eta}(\mathrm{Tr}(a^2))=j\}, ~~j\in \{1,-1\} .
\end{eqnarray*}
Then, the desired results follow from Lemma~\ref{lem:ti}.

%For $m$ being even, we conclude by Lemma~\ref{lm: solution of quadra form} that
%\begin{eqnarray*}
%     n_{1,1}&=&\frac{1}{2}\#\{a\in \mathbb{F}_{r}^*:\bar{\eta}(\mathrm{Tr}(a^2))=1\}\\
%            &=&\frac{p-1}{4}\left(p^{m-1}-(-1)^{\frac{p-1}{2}\frac{m}{2}}p^{\frac{m-2}{2}}\right),\\
%   n_{1,-1} &=&\frac{p-1}{4}\left(p^{m-1}-(-1)^{\frac{p-1}{2}\frac{m}{2}}p^{\frac{m-2}{2}}\right).
% \end{eqnarray*}
 \hfill\space$\qed$
\end{proof}

\begin{lemma}\label{lem:si}
For any $a\in \mathbb{F}_{r}^* $, let
\begin{equation*}
    s_i=\#\{a\in \mathbb{F}_{r}^*:\eta(a)\bar{\eta}(\mathrm{Tr}(a^{-1}))=i \}, ~~i\in\{1,-1\}.
\end{equation*}
Then, for $m$ being even, we have
\begin{eqnarray*}
     s_1=   s_{-1}=\frac{p^{m}-p^{m-1}}{2},
 \end{eqnarray*}
 and for $m$ being odd, we have
\begin{eqnarray*}
  \left\{\begin{array}{lll}
   s_1&=&\frac{p-1}{2}\left(p^{m-1}+(-1)^{\frac{p-1}{2}\frac{m-1}{2}}p^{\frac{m-1}{2}}\right),\\
   s_{-1}&=&\frac{p-1}{2}\left(p^{m-1}-(-1)^{\frac{p-1}{2}\frac{m-1}{2}}p^{\frac{m-1}{2}}\right).\\
\end{array} \right.
 \end{eqnarray*}
\end{lemma}
\begin{proof}
Since $\bar{\eta}$ and $\mathrm{Tr}$ are balanced, we have
\begin{eqnarray*}
   &&\#\{a\in \mathbb{F}_{r}^*:\bar{\eta}(\mathrm{Tr}(a^{-1}))=-1 \} \\
       &&=\frac{1}{2}\#\{a\in \mathbb{F}_{r}^*:\mathrm{Tr}(a^{-1})\neq 0 \} \\
       &&=\frac{p^{m}-p^{m-1}}{2}.
\end{eqnarray*}
Following from the definition above, we deduce that
\begin{eqnarray*}
    s_1&=&\#\{a\in \mathbb{F}_{r}^*:\eta(a)\bar{\eta}(\mathrm{Tr}(a^{-1}))=1 \}\\
       &=&n_{1,1}+n_{-1,-1}\\
       &=&n_{1,1}+\#\{a\in \mathbb{F}_{r}^*:\bar{\eta}(\mathrm{Tr}(a^{-1}))=-1 \}-n_{1,-1}\\
       &=&n_{1,1}-n_{1,-1}+\frac{p^{m}-p^{m-1}}{2}.
\end{eqnarray*}
Similarly,
\begin{eqnarray*}
 s_{-1}&=&\#\{a\in \mathbb{F}_{r}^*:\eta(a)\bar{\eta}(\mathrm{Tr}(a^{-1}))=-1 \}~~~~~~~~~~~~~\\
          &=&n_{1,-1}-n_{1,1}+\frac{p^{m}-p^{m-1}}{2}.
\end{eqnarray*}

The asserted results then follow from Lemma~\ref{lem:nij}.\hfill\space$\qed$
\end{proof}

{\begin{remark} Lemma~\ref{lem:si} determines $s_1$ and $s_{-1}$ with the help of Lemma~\ref{lem:nij}. In the following, we give another short proof for the case of $m$ being even.

If $m$ is even, Lemma~\ref{le:eta} states that $\eta(y) = 1$ for each $y \in \mathbb{F}_{p}^*$.
And there exists a $y' \in \mathbb{F}_{p}^*$ such that $\bar{\eta}(y') = -1$. So $\eta(y')\bar{\eta}(y') = -1$.

It follows that
\begin{eqnarray*}
    s_1&=&\#\{a\in \mathbb{F}_{r}^*:\eta(a)\bar{\eta}(\mathrm{Tr}(a^{-1}))=1 \}\\
       &=&\#\{a\in \mathbb{F}_{r}^*:\eta(a)\bar{\eta}(\mathrm{Tr}(a))=1 \}\\
       &=&\#\{y'a\in \mathbb{F}_{r}^*:\eta(y'a)\bar{\eta}(\mathrm{Tr}(y'a))=1 \}\\
       &=&\#\{y'a\in \mathbb{F}_{r}^*:\eta(y')\bar{\eta}(y')\eta(a)\bar{\eta}(\mathrm{Tr}(a))=1 \}\\
       &=&\#\{a\in \mathbb{F}_{r}^*:\eta(a)\bar{\eta}(\mathrm{Tr}(a))=-1 \}= s_{-1}.
\end{eqnarray*}
Hence we have
$    s_1=   s_{-1}=\frac{p^{m}-p^{m-1}}{2} $ since
{
\begin{eqnarray*}
    s_1+s_{-1}=\#\{a\in \mathbb{F}_{r}^*:\mathrm{Tr}(a^{-1})\neq 0 \}
              =p^m-p^{m-1}.
\end{eqnarray*}
}

 \end{remark}

}

\subsection{The proof of Theorem~\ref{thcwe:CD}}\label{sec:proofthCD}

It is now sufficient to show the complete weight enumerator of $C_D$ as stated in Theorem~\ref{thcwe:CD}.
Recall that
\begin{equation*}
    C_D=\{(\mathrm{Tr}(ax^2))_{x\in D}:a\in \mathbb{F}_{r}\},
\end{equation*} where $ D=\{x\in \mathbb{F}_{r}^*:\mathrm{Tr}(x)=0\}$.

It is { obvious} that the code $C_D$ has length $n=p^{ m-1}-1$ and dimension $m$.

Observe that $a=0$ gives the zero codeword and the contribution to the complete
weight enumerator is $w_0^{n}.$ Hence, we assume that $a\in\mathbb{F}_{r}^*$ for the rest of the proof.

For each codeword $(\mathrm{Tr}(ax^2))_{x\in D}$ of $C_D$ and $\rho\in \mathbb{F}_{p}$, we should consider the number of solutions $x\in \mathbb{F}_{r}^*$ satisfying $\mathrm{Tr}(x)=0$ and $\mathrm{Tr}(ax^2)=\rho$, i.e.,
\begin{eqnarray*}
N_a(\rho)=\#\{x\in \mathbb{F}_{r}^*:\mathrm{Tr}(x)=0 ~~ \mathrm{and}~~ \mathrm{Tr}(ax^2)=\rho \}.
\end{eqnarray*}

Note that $N_a(0)=p^{m-1}-1-\sum_{\rho\in \mathbb{F}_{p}^* }N_a(\rho)$. Thus we only need to calculate
$N_a(\rho)$ with $\rho\in \mathbb{F}_{p}^*$, which is shown in Lemma~\ref{lem:Na rho}.

The desired conclusions of Theorem~\ref{thcwe:CD}
then follow from Lemmas~\ref{lem:Na rho},~\ref{lem:ni} and~\ref{lem:si}.

{Note that for even $m$ and $\rho\in \mathbb{F}_{p}^*$, the frequencies of the codewords with $N_a(\rho)=p^{m-2}+(-1)^{\frac{p-1}{2}\frac{m}{2}}p^{\frac{m-2}{2}}$ and
$N_a(\rho)=p^{m-2}-(-1)^{\frac{p-1}{2}\frac{m}{2}}p^{\frac{m-2}{2}}$ are
$\frac{p^m-1}{2p}+\frac{p-1}{p}\eta_0^{(2,r)}$ and $\frac{p^m-1}{2p}+\frac{p-1}{p}\eta_1^{(2,r)}$,
respectively. The results are given in Table~\ref{CWE:m even simplified} according to Lemma~\ref{lemN=2}.}

\section{Minimal codewords in $C_D$}\label{sec:minimal}

In this section, we will show that each codeword of $C_D$ given by \eqref{def:CD} is minimal. Hence,
we need to introduce some definitions~\cite{ding2015twothree}.

The \emph{support} of a vector $\mathbf{c}=(c_0,\cdots,c_{n-1}) \in \mathbb{F}_{p}^n$ is defined as
\begin{eqnarray*}
\{0\leq i\leq n-1:c_i\neq 0\}.
\end{eqnarray*}
We say that a vector $\mathbf{x}$ covers a vector $\mathbf{y}$ if the support of $\mathbf{x}$ contains
that of $\mathbf{y}$ as a proper subset.

A \emph{minimal codeword} of a linear code $C$ is a nonzero codeword that
{does} not cover any other nonzero codeword of $C$.

If each nonzero codeword of $C$ is minimal, then the secret sharing scheme based on the dual code $C^{\perp}$ may have nice access structure, see Theorem $12$ of~\cite{ding2015twothree}. However, it is still very hard to construct such a linear code $C$.
% whose nonzero codewords are minimal.
We list the following lemma for minimal codewords~\cite{Ashikhmin1998minimal,Ashikhmin1995minimal}.

\begin{lemma}\label{lm:minimal}
Every nonzero codeword of a linear code $C$ over $\mathbb{F}_{p}$ is minimal, provided that
\begin{eqnarray*}
\frac{w_{min}}{w_{max}}>\frac{p-1}{p},
\end{eqnarray*}
where $w_{min}$ and $w_{max}$ denote the minimum and maximum nonzero weights in $C$, respectively.
\end{lemma}

For the linear codes $C_D$ of \eqref{def:CD}, if $m$ is even and $m\geq 4$, we have
\begin{eqnarray*}
\frac{w_{min}}{w_{max}}=\frac{p^{m-2}-p^{\frac{m-2}{2}}}{p^{m-2}+p^{\frac{m-2}{2}}}>\frac{p-1}{p}.
\end{eqnarray*}

The case of $m\geq 4$ being odd is proved similarly, see also~\cite{ding2015twothree}.

Using Lemma~\ref{lm:minimal}, we
{conclude} that all the nonzero codewords of $C_D$ are minimal if $m\geq 4$. Therefore, we can construct secret sharing schemes based on the dual codes $C_D^{\perp}$ with nice access structures. We omit the details here since it is similar to that of~\cite{ding2015twothree}.

\section{Concluding remarks}\label{sec:conclusion}

Inspired by~\cite{ding2015twodesign,ding2015twothree}, we constructed a family of three-weight
{linear codes.} Their complete weight enumerators and weight enumerators were presented explicitly in this paper. We also showed that every nonzero codeword of $C_D$ is minimal and thus the dual codes $C_D^{\perp}$ can be applied to construct secret sharing schemes.

\begin{acknowledgements}
%The authors are very grateful to the editor and the anonymous reviewers for their helpful comments and suggestions, which have improved the presentation of this paper.
The work of Zheng-An Yao is partially supported by the NSFC (Grant No.11271381), the NSFC (Grant No.11431015)
and China 973 Program (Grant No. 2011CB808000).
This work is also partially supported by the NSFC (Grant No. 61472457) and Guangdong Natural Science
Foundation (Grant No. 2014A030313161).
\end{acknowledgements}

%\bibliography{code}
%\bibliographystyle{spmpsci}

\end{document}